\documentclass[12pt]{article}
\usepackage{amsmath,amsthm,amsbsy,amsfonts}
\usepackage{bm,graphicx,psfrag,epsf}
\usepackage{enumerate}
\usepackage{natbib} 
\usepackage{url} 

\usepackage{bm}
\usepackage{algorithmic}
\usepackage{algorithm}
\usepackage{comment}

\usepackage{mathtools}
\usepackage{xcolor}

\newtheorem{proposition}{Proposition}[section]

\def\dist{\mathop{\rm dist}\nolimits}

\def\prox{\mathop{\rm prox}\nolimits}
\def\argmin{\mathop{\rm argmin}\nolimits}
\def\amp{\mathop{\;\:}\nolimits}
\newcommand{\bzero}{\boldsymbol{0}}

\newcommand{\bx}{\boldsymbol{x}}
\newcommand{\by}{\boldsymbol{y}}

\newcommand{\bD}{\boldsymbol{D}}

\newcommand{\bI}{\boldsymbol{I}}

\newcommand{\bW}{\boldsymbol{W}}
\newcommand{\bX}{\boldsymbol{X}}

\newcommand{\bbeta}{\boldsymbol{\beta}}

\newcommand{\bepsilon}{\boldsymbol{\epsilon}}

\newcommand{\btheta}{\boldsymbol{\theta}}

\newcommand{\Real}{\mathbb{R}}

\newcommand{\Tra}{^{\sf T}} 
\newcommand{\Inv}{^{-1}} 

\newcommand{\V}[1]{{\bm{\mathbf{\MakeLowercase{#1}}}}} 

\newcommand{\blind}{0}

\begin{document}

\def\spacingset#1{\renewcommand{\baselinestretch}%
{#1}\small\normalsize} \spacingset{1}


\if0\blind
{
  \title{\bf A Sharper Computational Tool for $\text{L}_2\text{E}$ Regression}
  \author{Xiaoqian Liu\\
    Department of Statistics, North Carolina State University\\
    Eric C. Chi\\
    Department of Statistics, 
    Rice University\\
     and \\
    Kenneth Lange\\ 
    Departments of Computational Medicine, Human Genetics, and Statistics\\
    University of California, Los Angeles}
    \date{}
  \maketitle
} \fi

\if1\blind
{ 
  \bigskip
  \bigskip
  \bigskip
  \begin{center}
    {\LARGE\bf A Sharper Computational Tool for $\text{L}_2\text{E}$ Regression}
\end{center}
  \medskip
} \fi

\bigskip
\begin{abstract}

Building on previous research of \cite{L2E}, the current paper revisits estimation in robust structured regression under the $\text{L}_2\text{E}$ criterion. We adopt the majorization-minimization (MM) principle to design a new algorithm for updating the vector of regression coefficients. Our sharp majorization achieves faster convergence than the previous alternating proximal gradient descent algorithm \citep{L2E}. In addition, we reparameterize the model by substituting precision for scale and estimate precision via a modified Newton's method. This simplifies and accelerates overall estimation. We also introduce distance-to-set penalties to enable  constrained estimation under nonconvex constraint sets. This tactic also improves performance in coefficient estimation and structure recovery. Finally, we demonstrate the merits of our improved tactics through a rich set of simulation examples and a real data application. 

\end{abstract}

\noindent%
{\it Keywords:} Integral squared error criterion; MM principle; Newton's method; penalized estimation; distance penalization
\vfill

\newpage
\spacingset{2} 

\section{Introduction}
\label{sec:intro}

Linear least squares regression quantifies the relationship between a response and a set of predictors. As such, it has been the most popular and productive technique of classical statistics.  The growing complexity of modern datasets necessitates special structures on the vector of regression coefficients. A typical example is sparse regression for high-dimensional data, where the number of predictors exceeds the number of responses. In this setting, assuming the coefficient vector is sparse not only improves a regression model's interpretability but also improves its prediction accuracy. The most popular vehicle for dealing with sparse regression is the least absolute shrinkage and selection operator (Lasso) \citep{Lasso}. Other examples of structured regression include isotonic regression \citep{isotonic}, convex regression \citep{cvxreg}, and ridge regression \citep{ridge}. 
 
Traditional structured regression estimates parameters by constrained least squares. Unfortunately, least squares estimates are extremely sensitive to outliers. A single outlier can ruin estimation accuracy. Consequently, robust structured regression has gained considerable traction in recent years. Numerous authors have contributed to the current body of techniques. To mention a few, \cite{robust-isotonic} propose a family of robust estimates for isotonic regression that replaces the least squares criterion with the M-estimation criterion \citep{M-estimate}. \cite{robust-convex} employ absolute error loss in robust convex regression.  This is also  an instance of an M-estimator. \cite{Extended-Lasso} suggest an extended Lasso method incorporating a stochastic noise term to account for corrupted observations in robust sparse multiple regression. \cite{sparse-trim} add a Lasso penalty to the least trimmed squares (LTS) loss to produce a robust sparse estimator that trims outliers by effectively minimizing the sum of squared residuals over a selected subset.
\cite{MD-lasso} adopt the minimum distance criterion to design a log-scaled loss function and propose the  minimum distance Lasso method for robust sparse regression. Other robust sparse regression methods can be found in \cite{LAD-Lasso, She2011, ES-Lassp}. 

The above works investigate robust structured regression on a case-by-case basis. \cite{trim-general} develop a family of trimmed regularized M-estimators with a wider focus but with the need to select the degree of trimming. Recently, \cite{L2E} derive yet another general framework for robust structured regression that simultaneously estimates regression coefficients as well as a precision parameter, which plays the same role as the trimming parameter in \cite{trim-general}.  \cite{L2E} use the $\text{L}_2\text{E}$ criterion \citep{scott1992} to quantify goodness-of-fit and a convex penalty to enforce structure.  Their algorithmic  framework solves the corresponding  optimization problem by block descent. Although the computational framework presented in \cite{L2E} is general, there is room for some nontrivial improvements. First, the proposed proximal gradient algorithm for updating both the regression coefficients and the precision parameter at each block descent iteration can be slow to converge. Second, the box constraint on the precision parameter introduces two additional hyper-parameters that must be specified. Finally, while \cite{L2E} focused on convex penalties and constraints, the framework that they introduced is not inherently limited to convex options and warrants extension to important nonconvex alternatives that impose desirable structures.

The limitations in \cite{L2E} just discussed motivate the current paper and its new contributions. First, we derive a majorization-minimization algorithm to accelerate the estimation of the regression coefficients. Second, we reparameterize the precision parameter to eliminate the box constraint. A simple one-dimensional approximate Newton's method quickly solves the resulting smooth unconstrained problem for updating precision. Finally, we demonstrate improved statistical performance by imposing nonconvex penalties. Specifically, we adopt distance-to-set penalties to improve estimation accuracy subject to structural constraints. These improvements do not compromise robustness.

The rest of this paper is organized as follows. In Section \ref{Sec2}, we review the $\text{L}_2\text{E}$ criterion, the majorization-minimization (MM) principle, and distance penalization. In Section \ref{Sec3}, we set up the optimization problem for robust structured regression under the $\text{L}_2\text{E}$ criterion. In Section \ref{Sec4}, we introduce strategies that improve the estimation techniques of \cite{L2E}. In Sections \ref{Sec5} and \ref{real_data}, we provide a rich set of simulation examples and a real data application to demonstrate the empirical performance of our new algorithms. We end with a discussion in Section \ref{Sec6}.

\section{Background}
\label{Sec2}

\subsection{The $\text{L}_2\text{E}$ Criterion}

Although traditionally used in nonparametric estimation, the $\text{L}_2\text{E}$ criterion, also known as the integrated squared error (ISE), can be exploited in parametric settings for robust estimation. Suppose the goal is to estimate a density function $f(\bx \mid \btheta)$, where the true parameter $\btheta_*$ is unknown. The $\text{L}_2\text{E}$ criterion seeks to estimate $\btheta$ by minimizing the $L_2$ distance between $f(\bx \mid \btheta)$ and $f(\bx \mid \btheta_*)$; thus
\begin{eqnarray}
    \hat{\btheta} & = & \argmin_{ \btheta} \int \left[ f(\bx \mid \btheta) - f(\bx \mid \btheta_*)\right]^2 \,d\bx \notag \\
    & = &\argmin_{\btheta}  \int f(\bx \mid \btheta)^2 \,d\bx - 2 \int f(\bx \mid \btheta) f(\bx \mid \btheta_*) \,d\bx + \int f(\bx \mid \btheta_*)^2 \,d\bx.   \label{L2-f}
\end{eqnarray}
The third integral in formula \eqref{L2-f} does not depend on $\btheta$ and can be excluded from the minimization. The second integral is the expectation of $f(\bx \mid \btheta)$ and can be approximated by an unbiased estimate, namely its sample mean. Therefore, an approximate  $L_2$ estimate of $\btheta$  is
\begin{eqnarray}
  \label{L2E-f}
    \hat{\btheta}_{\text{L}_2\text{E}} & = & \argmin_{\btheta}  
    \int f(\bx \mid \btheta)^2 \, d\bx - \frac{2}{n} \sum_{i=1}^n f(\bx_i\mid\btheta),
\end{eqnarray}
where $n$ denotes the sample size. The $\text{L}_2\text{E}$ represents  a trade-off between efficiency and robustness.  It is less efficient but more robust than the maximum likelihood estimate (MLE) \citep{L2E-Scott, warwick2005}. \cite{L2E} discuss in detail how the $\text{L}_2\text{E}$ estimator imparts robustness in  structured regression.

\subsection{The MM Principle}

The majorization-minimization principle \citep{lange2000optimization,lange2016mm} for minimizing an objective function $h(\btheta)$ involves two steps, a) majorization of $h(\btheta)$ by a surrogate function $g(\btheta \mid \btheta_k)$ anchored at the current iterate $\btheta_k$ and then b) minimization of $\btheta \mapsto g(\btheta \mid \btheta_k)$ to construct $\btheta_{k+1}$. The surrogate function $g(\btheta \mid \btheta_k)$ must satisfy the two requirements:
 \begin{eqnarray}
     h(\btheta_k) & = & g(\btheta_k \mid \btheta_k), \quad  \text{tangency}  \label{MM-cond1}\\
     h(\btheta) & \leq & g(\btheta \mid \btheta_k) \quad \text{for all}\, \: \btheta, \quad \text{domination}.  \label{MM-cond2}
 \end{eqnarray}
Under these conditions, the iterates enjoy the descent property $h(\btheta_{k+1}) \leq h(\btheta_k)$ as demonstrated by the relations 
\begin{eqnarray*}
h(\btheta_{k+1}) & \leq & g(\btheta_{k+1}\mid \btheta_k) \amp \leq \amp g(\btheta_k \mid \btheta_k) \amp = \amp h(\bm \theta_k),
\end{eqnarray*}
reflecting conditions \eqref{MM-cond1} and \eqref{MM-cond2}. Ideally, the MM principle converts a hard optimization problem into a sequence of easier ones. The key to success is the construction of a tight majorization that can be easily minimized. In some problems it is possible to construct a sharp majorization within a limited class of majorizers.  Figure \ref{fig: MM} depicts a sharp quadratic majorization that is best among all quadratic majorizations that share the same tangency point. Sharp majorization accelerates the convergence of a derived MM algorithm \citep{sharp-mm}. In practice, majorization can be done piecemeal by exploiting the convexity or concavity of the various terms comprising the objective. 
 
\begin{figure}[t]
\centering
{\includegraphics[width=10cm]{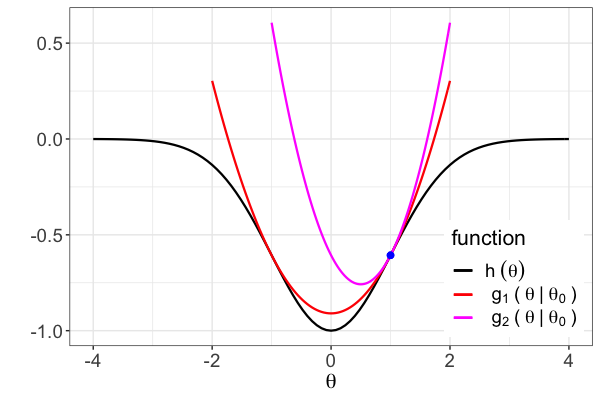}
}
\setlength{\abovecaptionskip}{-10pt} 
\caption{An example of sharp quadratic majorization. The quadratic $g_1(\theta \mid \theta_0)$ offers the sharpest majorization of the loss $h(\theta)$ and falls below every looser quadratic majorization $g_2(\theta \mid \theta_0)$.}
\label{fig: MM}
\end{figure}

\subsection{Distance Penalization}
\label{distance-penalization}

To estimate a parameter vector $\btheta$ subject to a set constraint $\btheta \in C$, it is often convenient to employ a squared Euclidean distance penalty \citep{dist-majorization, dist-to-set}.  For a closed set $C$, the penalty is defined as
\begin{eqnarray}
    \label{distance-penalty}
\frac{1}{2} \dist(\btheta, C)^2 & = & \min_{\bbeta \in C}\frac{1}{2}\|\btheta-\bbeta\|_2^2.
\end{eqnarray}
The beauty of this penalty is that it is majorized at the current iterate $\btheta_k$ by the spherical quadratic
\begin{eqnarray}
    \label{dist-penalty}
\frac{1}{2}\|\btheta - \mathcal{P}_C(\btheta_k)\|_2^2,
\end{eqnarray}
 where $\mathcal{P}_C(\btheta)$ denotes the Euclidean projection of $\btheta$ onto $C$ \citep{bauschke2011convex}. When $C$ is both closed and convex, $\mathcal{P}_C(\btheta)$ consists of a single point. For nonconvex sets, $\mathcal{P}_C(\btheta)$ sometimes consists of multiple points. When $\mathcal{P}_C(\btheta)$ is single valued, the distance penalty (\ref{distance-penalty}) has gradient $\btheta - \mathcal{P}_C(\btheta)$ 

The proximal distance method of constrained optimization minimizes the penalized objective $h(\btheta)+\frac{\rho}{2}\dist(\btheta, C)^2$ \citep{dist-to-set,ProxDist}. The tuning constant $\rho$ controls the trade-off between minimizing the loss $h(\btheta)$ and satisfying the constraint $\btheta \in C$. Under suitable  regularity conditions, the constrained solution can be recovered in the limit as $\rho$ tends towards infinity  \citep{dist-majorization, ProxDist}.
Therefore, a large value of $\rho$, say $10^8$, is chosen in practice to enforce the constraint. The MM principle suggests majorizing the distance penalty by the spherical quadratic (\ref{dist-penalty}) and applying the proximal map $\btheta_{k+1}=\prox_{\rho^{-1}h}[\mathcal{P}_C(\btheta_k)]$ to generate the next iterate. The proximal distance principle applies to a wide array of models, including sparse regression, nonnegative regression, and low-rank matrix completion. It is accurate in estimation  and avoids the severe shrinkage of Lasso penalization with well-behaved constraint sets \citep{dist-to-set}. \cite{fused-dist} extend distance penalization to fusion constraints of the form $\bD \bbeta \in C$ involving a fusion matrix $\bD$ such as a discrete difference operator. Although the advantages of proximal maps are lost, this extension brings more constrained statistical models under the umbrella of distance penalization. 
\section{$\text{L}_2\text{E}$ Robust Structured Regression}
\label{Sec3}

Consider the classical linear regression model $\by=\bX\bbeta+\tau^{-1}\bepsilon$, where $\by \in \Real^n$ is the response vector, $\bX \in \Real^{n \times p}$ is the design matrix of predictors, and $\bepsilon \in \Real^n$ is the noise vector with independent standard Gaussian components. The regression coefficients $\bbeta \in \Real^p$ and the precision $\tau \in \Real_+$ are the parameters of the model. Collectively, we denote the parameters by $\btheta=(\bm \beta\Tra, \tau)\Tra$. The density of the $i$th response $y_i$ amounts to
\begin{eqnarray*}
    \label{ydensity}
    f (y_i \mid \btheta) & = & \frac{\tau}{\sqrt{2\pi}} e^{-\frac{\tau^2 r_i^2}{2}},
\end{eqnarray*}
where $r_i = y_i - \bm x_i \Tra \bm \beta$ is the $i$th residual. A brief calculation shows that equation \eqref{L2E-f} gives rise to the $\text{L}_2\text{E}$ loss 
\begin{eqnarray}
\label{L2E-loss}
    h(\btheta) &= &\int f\left(y \mid \btheta\right)^2 d y-\frac{2}{n} \sum_{i=1}^n f\left(y_{i}\mid \btheta\right)
    \amp = \amp \frac{\tau}{2 \sqrt{\pi}}-\frac{\tau}{n} \sqrt{\frac{2}{\pi}} \sum_{i=1}^{n} e^{-\frac{\tau^{2}r_{i}^{2}}{2}}.
\end{eqnarray}

Structured regression introduces set constraints on the regression coefficient vector $\bbeta$. Consequently, $\text{L}_2\text{E}$ aims to solve the constrained optimization problem
\begin{eqnarray}
\label{objective1}
    \min_{\bbeta \in \Real^p, \tau\in \Real_+} h(\bbeta, \tau), ~~~~\text{subject to}~~
    \bm \beta \in C.
\end{eqnarray}
For example, $C=\{\bbeta \in \Real^p :\beta_1 \leq \beta_2 \leq \cdots \leq  \beta_p\}$ leads to a robust isotonic regression problem. Sparsity can be imposed directly by taking $C=\{\bm \beta \in \Real^p : \lVert \bm \beta \lVert_0 \leq k\}$ for some positive integer $k$ or indirectly by taking $C=\{\bm \beta \in \Real^p : \lVert \bm \beta \lVert_1 \leq t\}$ for $t>0$. Alternatively, we can rewrite  problem \eqref{objective1} as the non-smooth optimization problem
\begin{eqnarray}
\label{objective2}
      \min_{\bm \beta \in \Real^p, \tau\in \Real_+} h(\bm \beta, \tau) +
      \phi(\bbeta),
\end{eqnarray}
where the penalty $\phi(\bbeta)$ is either the 0/$\infty$ indicator of the constraint set $C$ denoted by $\iota_C(\bm \beta)$ or a better behaved but still non-smooth substitute such as the Lasso. Although we emphasize structured regression, the formulations \eqref{objective1} and \eqref{objective2} also include unstructured multivariate regression where $C=\Real^p$ and $\phi(\V \beta) \equiv 0$.

Solving problem \eqref{objective1}, or equivalently solving \eqref{objective2}, is challenging for two reasons. First, both problems  are nonconvex owing to the nonconvexity of the $\text{L}_2\text{E}$ loss \eqref{L2E-loss}. Second, the penalty term $\phi(\bm \beta)$ may be  non-differentiable. Fortunately, the block gradients of the $\text{L}_2\text{E}$ loss with respect to $\bbeta$ and $\tau$,  $\nabla_{\bbeta} h(\bbeta, \tau)$ and $\frac{\partial}{\partial \tau}h(\bbeta, \tau)$, are Lipschitz. This key property motivates a block descent algorithm \citep{L2E} that alternates between reducing the objective with respect to $\bbeta$ and $\tau$, holding the other block fixed. \cite{L2E} also impose the bounds $0 <\tau_{\min} \le \tau \le \tau_{\max}<\infty$ on $\tau$.

An appealing property of block descent is that the objective function  is guaranteed to decrease at each iteration.  \cite{L2E} apply proximal gradient descent to decrease the objective in each block update. Because the proximal gradient updates are based on a loose loss majorization, the algorithm is slow to converge.  To ameliorate this fault, we propose new strategies for updating $\bbeta$ and $\tau$ in the next section.

\section{Computational Methods}
\label{Sec4}

\subsection{Updating the Regression Coefficients}

Consider the problem of updating the regression coefficients $\bbeta$. Because the contribution  $-e^{-\tau^2r_i^2/2}$ to the $\text{L}_2\text{E}$ loss \eqref{L2E-loss} is differentiable and concave with respective to $r_i^2$, we can exploit the concave majorization 
\begin{eqnarray*}
\label{concavity}
    f(u) &\leq & f(u_k) + f'(u_k)(u-u_k)
\end{eqnarray*}
in the form 
\begin{eqnarray}
\label{majorization-r}
    -e^{-\tau^2r_i^2/2} &\leq& -e^{-\tau^2r_{ki}^2/2} +\frac{\tau^2}{2} e^{-\tau^2 r_{ki}^2/2}\left(r_i^2 - r_{ki}^2\right)
\end{eqnarray}
around the tangency point $r_{ki}^2$. By omitting irrelevant multiplicative and additive terms, this produces the surrogate function
\begin{eqnarray}
\label{beta-majorizer}
    f(\bbeta \mid \bbeta_k, \tau) & = & \frac{1}{2} \sum_{i=1}^n e^{-\tau^2 r_{ki}^2/2}(y_i-\bx_i\Tra \bbeta)^2
     \:\: = \:\: \frac{1}{2}\lVert \Tilde{\by} - \Tilde{\bX} \bbeta\lVert_2^2
\end{eqnarray}
for the $\text{L}_2\text{E}$ loss \eqref{L2E-loss}, where $r_{ki} = y_i - \bm x_i\Tra \bbeta_k$ is the $i$th residual at iteration $k$, $\tilde{\by} = \sqrt{\bW_k}\by$, $\tilde{\bX} = \sqrt{\bW_k}\bX$, and $\bW_k \in \Real^{n \times n}$ is a diagonal weight matrix with the $i$th diagonal entry $e^{-\tau^2 r_{ki}^2/2}$. 

The next proposition demonstrates that the surrogate (\ref{beta-majorizer}) is the sharpest quadratic majorization in the residual variables $r_i$. 
It does not claim that the majorization \eqref{beta-majorizer} is the sharpest multivariate quadratic majorization in the full variable $\bbeta$. Despite this fact, the majorization yields substantial gains in computational efficiency over the looser proximal gradient majorization pursued by \cite{L2E}.
\begin{proposition}
\label{prop1}
Let $f(r) = -e^{-a r^2}$ with $a>0$. Then the symmetric quadratic function $g(r) = -e^{-a r_k^2} + a e^{-a r^2} (r^2-r_k^2)$ is the sharp quadratic majorizer of $f(r)$. 
\end{proposition}
\begin{proof}
\cite{van2005} proves that a univariate quadratic function $g(r)$ majorizing a univariate differentiable function $f(r)$ and touching it at two points is  sharp. In the present case, $g(r)$ touches $f(r)$ at the points $r=\pm r_k$.
\end{proof}

For an $\text{L}_2\text{E}$ loss with penalty $\phi(\bbeta)$, the next MM iterate is 
\begin{eqnarray*}
\label{update-beta-MM}
  \bbeta_{k+1} & = & \argmin_{\bbeta \in \Real^p}~ \frac{1}{2}\lVert \Tilde{\by} - \Tilde{\bX} \bbeta\lVert_2^2+ \phi(\bbeta).
\end{eqnarray*}
In the setting of distance penalization with a fusion penalty, the surrogate reduces to the least squares criterion
\begin{eqnarray*}
\frac{1}{2}\left\| \begin{pmatrix}\Tilde{\by} \\ \sqrt{\rho}\mathcal{P}_C(\bD\bbeta_k)\end{pmatrix} - \begin{pmatrix}\Tilde{\bX} \\ \sqrt{\rho}\bD\end{pmatrix} \bbeta
\right\|_2^2,
\end{eqnarray*}
which is amenable to minimization by the QR algorithm or the conjugate gradient algorithm. The computational complexity of the $\bbeta$ update is dominated by this least squares problem. Indeed, computation of the current residuals, the matrix $\bW_k$, the product $\tilde{\by}$, and the product $\tilde{\bX}$ require, respectively, operation counts of $\mathcal{O}(np)$, $\mathcal{O}(n)$, $\mathcal{O}(n)$, and $\mathcal{O}(np)$. Updating $\bm \beta$ using proximal gradient descent requires similar steps. Evaluation of the proximal map of $\phi(\bbeta)$ reduces to penalized least squares with an identity design matrix. Hence, with a diagonal design matrix $\bX$, the computational cost per iteration of the current MM algorithm is essentially the same as that of the proximal gradient descent algorithm in \cite{L2E}. The numbers of iterations until convergence of the two algorithms are vastly different however. Additionally,  the distance penalized MM algorithm is more flexible in allowing nonconvex and fusion constraints.

\subsection{Updating the Precision Parameter}

There are two concerns in updating $\tau$, namely the slow convergence of proximal gradient descent and the presence of box constraints on $\tau$. To attack the latter concern, we reparameterize by setting $\tau = e^{\eta}$ for any real valued $\eta$. Because the stationary condition for minimizing the loss $h(\bbeta, e^{\eta})$ with respect to $\eta$ is intractable, we turn to a variant of Newton's method. The required first and second derivatives are
\begin{eqnarray*}
     \frac{\partial}{\partial \eta}h(\bm \beta, e^{\eta}) &=& \frac{e^{\eta}}{2\sqrt{\pi}}-\frac{e^{\eta}}{n}\sqrt{\frac{2}{\pi}}\sum_{i=1}^n w_i
     +\frac{e^{3\eta}}{n}\sqrt{\frac{2}{\pi}}\sum_{i=1}^n w_i r_{i}^{2} \\
     \frac{\partial^2}{\partial \eta^2}h(\bm \beta, e^{\eta})  &=& \frac{e^{\eta}}{2\sqrt{\pi}}+ \frac{4e^{3\eta}}{n}\sqrt{\frac{2}{\pi}}\sum_{i=1}^n w_ir_{i}^{2}
     -\frac{e^{\eta}}{n}\sqrt{\frac{2}{\pi}}\sum_{i=1}^n w_i
     -\frac{e^{5\eta}}{n}\sqrt{\frac{2}{\pi}}\sum_{i=1}^n w_ir_{i}^{4},
 \end{eqnarray*}
where $w_i =e^{-e^{2\eta} r_i^2/2}$ and $r_i$ is the $i$th residual. The Newton increment only points downhill when $\frac{\partial^2}{\partial \eta^2}h(\bm \beta, e^{\eta})$ is positive. This prompts discarding the negative contributions and relying on the approximation
\begin{eqnarray*}
\frac{\partial^2}{\partial \eta^2}h(\bbeta, e^{\eta}) & \approx &
d \amp=\amp \frac{e^{\eta}}{2\sqrt{\pi}}+ \frac{4e^{3\eta}}{n}\sqrt{\frac{2}{\pi}}\sum_{i=1}^n w_i r_{i}^{2}.
\end{eqnarray*}
Our modified Newton's iterates are defined by 
 \begin{eqnarray*}
     \label{newton-eta}
     \eta_{k+1} &=& \eta_k - t_kd_k^{-1} \frac{\partial}{\partial \eta}h(\bbeta, e^{\eta_k}),
 \end{eqnarray*}
where $t_k$ is a positive stepsize parameter chosen via Armijo backtracking started at $t_k=1$. Little backtracking is needed because replacing $\frac{\partial^2}{\partial \eta^2}h(\bbeta, e^{\eta})$ by the larger value $d$ diminishes the chances of overshooting the minimum of $h(\bbeta,e^\eta)$.

Our modified Newton's method enjoys the same computational complexity as proximal gradient descent. The dominant computational expense in updating $\eta$ in both algorithms  comes from computing the residuals $r_i$. This step requires $\mathcal{O}(np)$ operations. Once all $r_i$ are updated, computing the derivatives only requires an additional $\mathcal{O}(n)$ operations. In summary, our new strategy converges in fewer iterations, removes the box constraint on $\tau$, and enjoys the same computational cost per iteration as proximal gradient descent. 

Algorithm \ref{alg: BCD-MM} summarizes our algorithm for minimizing the penalized loss  \eqref{objective2}. As in \cite{L2E}, we set the maximum numbers of inner iterations for updating $\bbeta$ and $\eta$ to be $N_{\bbeta}$ and $N_{\eta}$, respectively, at each outer iteration. Extreme values $N_{\bbeta}$ and $N_{\eta}$ tend to slow overall convergence. In our simulation studies, we set $N_{\bbeta} = N_{\eta} = 100$. In the algorithm the notation $\bW_+$ signifies that $\bW$ depends on the previous inner iterate $\bbeta_+$. 

\begin{algorithm}[tbh]
\caption{Block descent with MM and approximate Newton for problem \eqref{objective2}}
\label{alg: BCD-MM}
 Initialize: $\bbeta_0 \in \Real^p$, $\tau_0 \in \Real_+$,  $N_{\bbeta}$, and $N_{\eta}$.
\baselineskip=12pt
 \begin{algorithmic}[1]
\FOR {$k = 1, 2, \cdots$}
	 \STATE $\bbeta^+ \leftarrow \bbeta_{k-1}$ 
	 \FOR {$i = 1, \cdots, N_{\bbeta}$}
	\STATE $\tilde{\by} = \sqrt{\bW_+}\by$ \\
	\STATE $\tilde{\bX} = \sqrt{\bW_+} \bX$ \\
	\STATE  $\displaystyle \bbeta^+ ~~=~~
    \argmin_{\bbeta \in \Real^p}~ \frac{1}{2}\lVert \Tilde{\by} - \Tilde{\bX} \bbeta\lVert_2^2
      + \phi(\bbeta)$
	\ENDFOR
	 \STATE $\bbeta_{k} \leftarrow \bbeta^+$
	 \STATE $\eta^+ \leftarrow \log(\tau_{k-1})$ 
	 \FOR{$i = 1, \cdots, N_{\eta}$}
	 \STATE 
	     $\displaystyle \eta^+ ~~=~~ \eta^+ - t_i d_i\Inv \frac{\partial}{\partial \eta}h(\bbeta_k, e^{\eta^+}) $
	 \ENDFOR
	 \STATE $\tau_k \leftarrow e^{\eta^+}$ 
\ENDFOR
\end{algorithmic} 
\end{algorithm}

We close this section by stressing the importance of the weight matrix  $\bW_+$ in the success of $\text{L}_2\text{E}$ regression. The diagonal entry $e^{-\tau^2 r_{+i}^2/2}$ of $\bW_+$ depends on the $i$th residual from the previous inner iterate $\bbeta_+$ and downweights case $i$ if its residual is large. The converged weights also conveniently flag outliers. We will exploit this bonus later in Section \ref{real_data}.

\section{Numerical Experiments}
\label{Sec5}

To compare the estimation accuracy and computational efficiency of Algorithm \ref{alg: BCD-MM} (abbreviated MM) and proximal gradient descent (abbreviated PG), we consider isotonic regression and convex regression. To highlight the advantages of distance penalization over competing model selection methods, we consider sparse regression and trend filtering. For the sake of brevity, we relegate two of the examples to the supplement. Readers wishing to implement our version of $\text{L}_2\text{E}$ regression should visit the eponymous \texttt{L2E} R package \citep{l2e_package} on the Comprehensive R Archive Network (CRAN).

\subsection{Robust Isotonic Regression}

Classical isotonic regression involves minimizing the least squares criterion 
\begin{eqnarray*}
\|\by-\bbeta\|_2^2 & = & \sum_{i=1}^n (y_i-\beta_i)^2
\end{eqnarray*}
subject to $\bbeta$ belonging to the set $C_1 =\{\bbeta \in \Real^n : \beta_1 \leq \cdots \leq \beta_n\}$. Independent standard normal errors are implicit in this formulation. Here the design matrix $\bX=\bI_n$, and the mean function of the model is monotonically increasing and piecewise constant. In the $\text{L}_2\text{E}$ version of the problem, we impose the $0/\infty$ penalty $\phi(\bbeta) = \iota_{C_1}(\bbeta)$. The MM update of $\bbeta$ succumbs to the \texttt{gpava} function in the \texttt{isotone} R package \citep{iso-package}. As mentioned earlier, the MM $\bbeta$ update enjoys the same per-iteration computational cost as the PG $\bbeta$ update \citep{L2E}. 

In our simulation, $1000$ responses are generated by sampling points $x_i$ evenly from $[-2.5,2.5]$ and setting $ y_i = x_i^3 + s_i + \varepsilon_i$, where the $\varepsilon_i$ are i.i.d.\@ standard normal deviates, and the $s_i$ shift the underlying cubic signal. The responses define mean vector  $\bbeta \in \Real^{1000}$. Outliers are introduced at consecutive responses by setting $s_i = 14$ for $i=251, 252, \cdots, 250+m$, where $m$ is the number of outliers; all other responses have $s_i = 0$. The shift of 14 makes the contaminated responses match the maximum observed value in the uncontaminated responses. Each method is tested over 100 replicates and initialized by $\bbeta_0 = \by$, $\tau_0 = \text{MAD}(\by)^{-1}$ for PG, and $\eta_0 = -\log[\text{MAD}(\by)]$ for MM, where $\text{MAD}(\by)$ is the reciprocal of the median absolute deviation of the responses.

\begin{figure}[t]
\centering
{\includegraphics[width=7cm]{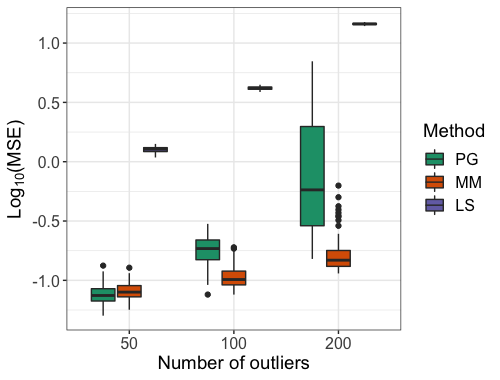}
}
\quad
{\includegraphics[width=7cm]{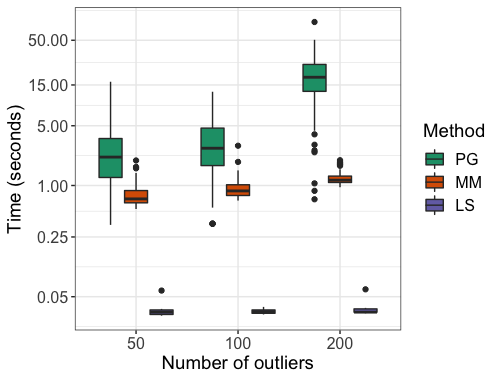}
  \setlength{\abovecaptionskip}{-10pt}
\caption{Simulation results for isotonic regression under different numbers of outliers. Boxplots depict the MSE (left panel) and run time (right panel) over $100$ replicates. 
}
\label{fig: isotonic}}
\end{figure}

Figure \ref{fig: isotonic} displays boxplots of the MSEs and run times in seconds in fitting the isotonic regression model under different numbers of outliers. We include the results from ordinary least squares (abbreviated LS) as a baseline. As anticipated, the estimation accuracy of LS degrades as the number of outliers increases. In contrast, both MM and PG exhibit much more modest increases in estimation error, with MM less sensitive to outliers than PG. Note that the optimization  problems of PG and MM differ slightly. We put a box constraint on $\tau$ for PG but reparameterize $\tau$ as $\tau=e^{\eta}$ for MM to eliminate the box constraint on $\tau$. For sufficiently large box constraints, the solutions to the two problems coincide, but differences in the algorithms will still  produce different algorithm iterate trajectories. As discussed in Section \ref{Sec3}, the $\text{L}_2\text{E}$ optimization problem is nonconvex and may exhibit multiple local minima. Thus, PG and MM may converge to different minima and produce different MSEs.

\begin{figure}[t]
\centering
{\includegraphics[width = 18cm]{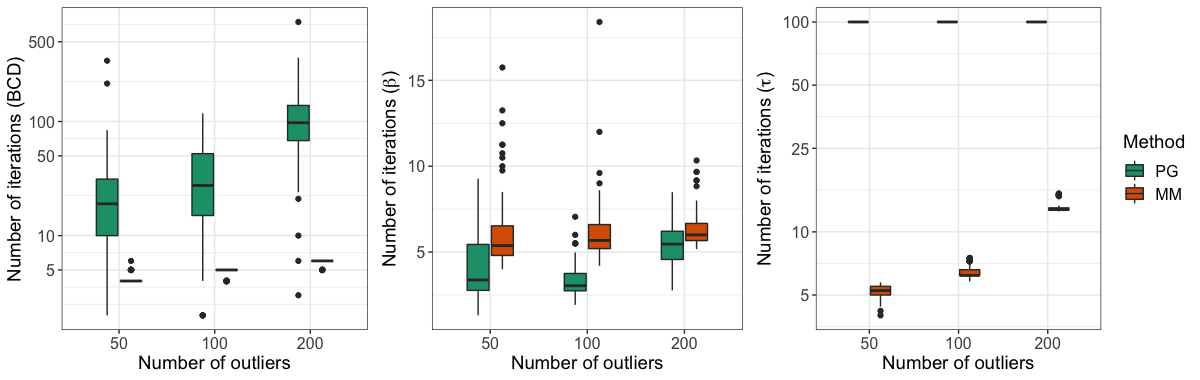}}
\setlength{\abovecaptionskip}{-25pt}
\caption{Boxplots of the mean number of outer block descent iterations (left panel), the mean number of inner iterations for updating $\bbeta$ per outer iteration (middle panel), and  the mean number of inner iterations for updating $\tau$ per outer iteration (right panel). All plots refer to the experiment summarized in Figure \ref{fig: isotonic}. 
}
\label{fig: isotonic-2}
\end{figure}

The right panel of Figure \ref{fig: isotonic} shows the significant speed advantage of MM over PG. Run times of PG increase rapidly as the number of outliers increases, while run times of MM are far more stable against the number of outliers. MM is less computationally efficient than LS, which avoids computation of case weights. The difference in run time between PG and MM is directly attributable to MM's reduced number of outer iterations until convergence. For the same experiment, Figure \ref{fig: isotonic-2} depicts boxplots of the mean number of outer block descent iterations, the mean number of inner iterations for updating $\bbeta$ per outer iteration, and the mean number of inner iterations for updating $\tau$ per outer iteration. Note that in our implementation, we terminate the inner iterations for updating $\bbeta$ and $\tau$ if certain convergence conditions are satisfied. Readers may refer to the \texttt{L2E} package for details. It may seem paradoxical that PG takes fewer inner iterations than MM to update $\bbeta$. However, recall that PG is fitting a less snug surrogate than MM. PG also takes far more inner iterations than MM to update $\tau$. This reflects the speed of our approximate Newton method.

The robust isotonic simulations also illustrate the ability of $\text{L}_2\text{E}$ regression to handle outliers under various contamination levels. To explore this tendency, we fix the number of outliers at $m=100$, vary the shifts $s_i$ over the grid $\{2, 5, 8, 14, 20\}$,  adopt the same initialization as the previous experiment, and run $100$ replicates for each scenario. Figure \ref{fig: isotonic-shifts} summarizes the estimation and computation performance of PG, MM, and LS under different contamination levels. When the data are only slightly contaminated ($s_i=2$), the two robust methods, PG and MM, fail to detect the outliers and achieve estimation accuracy comparable to LS. However, as the level of contamination $s_i$ grows, the MSE of LS increases rapidly, while the MSE of MM behaves robustly and quickly declines. Interestingly, the MSE of PG decreases gradually as the shift grows. These results suggest that both PG and MM need a certain level of contamination to successfully detect outliers. MM is more responsive to the contamination than PG even if the data are modestly contaminated. This is yet another advantage of MM over PG.

\begin{figure}[t]
\centering
{\includegraphics[width=7.5cm]{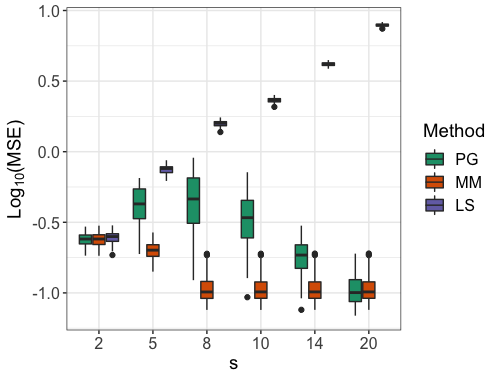}
}
\quad
{\includegraphics[width=7.5cm]{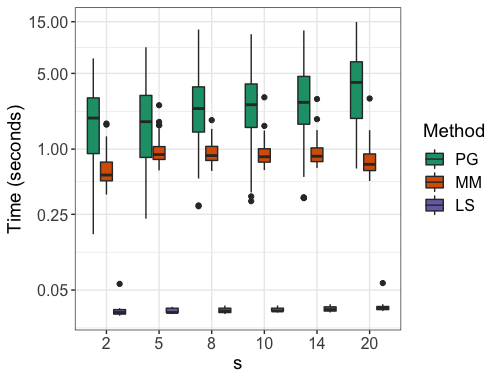}
 \setlength{\abovecaptionskip}{-10pt}
\caption{Simulation results for isotonic regression under different contamination  levels. Boxplots depict the MSE (left panel) and run time (right panel) over $100$ replicates.
}
\label{fig: isotonic-shifts}}
\end{figure}

\begin{figure}[t]
\centering
{\includegraphics[width = 18cm]{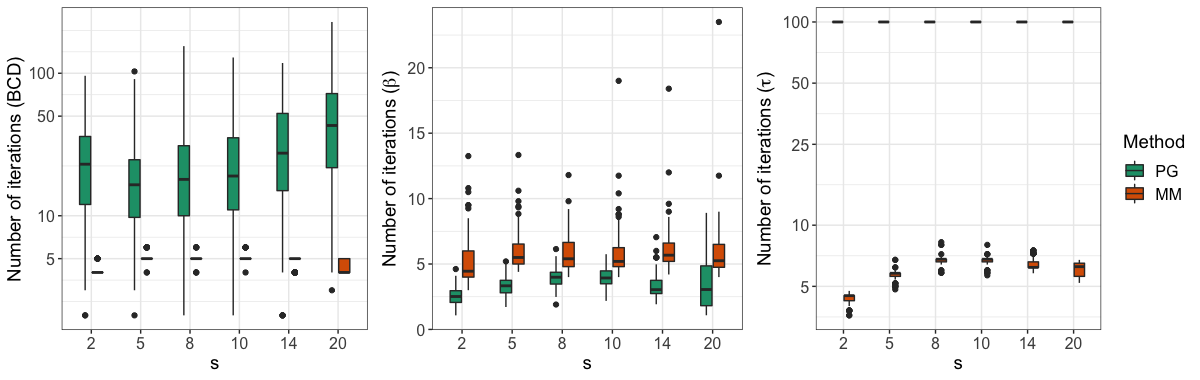}}
  \setlength{\abovecaptionskip}{-38pt}
\caption{Boxplots of the mean number of outer block descent iterations (left panel), the mean number of inner iterations for updating $\bbeta$ per outer iteration (middle panel), and  the mean number of inner iterations for updating $\tau$ per outer iteration (right panel). 
All plots refer to the experiment summarized in Figure \ref{fig: isotonic-shifts}. 
}
\label{fig: isotonic-shifts2}
\end{figure}

The right panel of Figure \ref{fig: isotonic-shifts} illustrates how PG's run times increase as the contamination level increases. The run times of MM, however, are stable with contamination level and consistently shorter than those of PG, though longer than those of LS.  Figure \ref{fig: isotonic-shifts2} explains the difference in the computational performance between PG and MM. The numbers of inner iterations for updating $\bbeta$ and $\tau$ for both PG and MM are insensitive to contamination level.  MM's number of outer block descent iterations is always small, while PG's number of outer iterations increases. This difference explains the speed advantage of MM.

\subsection{Robust Sparse Regression}

Sparse linear regression minimizes the penalized least squares criterion
\begin{eqnarray*}
\frac{1}{2}\|\by-\bm X \bbeta\|_2^2 + \phi(\bbeta),
\end{eqnarray*}
with $\phi(\bbeta)$ promoting sparsity. Typical choices of $\phi(\bbeta)$ includes the Lasso and the nonconvex MCP penalty \citep{MCP}. In the $\text{L}_2\text{E}$ framework, each MM update solves  a $\phi$-penalized least squares problem. The \texttt{ncvfit} function in the R package \texttt{ncvreg} is ideal for this purpose \citep{ncvreg}. In the distance penalty context, the constraint set is  $C_2 = \{ \bbeta \in \Real^p: \lVert \bbeta \lVert_0 \leq k\}$, where the positive integer $k$ encodes the sparsity level. The MM update of $\bbeta$ relies on the proximal distance principle and reduces to least squares.

To shed light on the statistical performance of $\text{L}_2\text{E}$ regression with Lasso, MCP, and distance penalties, we undertake a small simulation study involving a sparse coefficient vector $\bbeta=(1, 1, 1, 1, 1,0, \cdots, 0)^{\Tra} \in \Real^{50}$ and a design matrix $\bX \in \Real^{200 \times 50}$ whose independent entries are standard Gaussian deviates. The response $\bm y$ is simulated as $\bm y = \bm X \bm \beta + \bm \epsilon$ where components of $\bm \epsilon$ are standard normal noises. We then shift the first $m$ entries of $\by$ and the first $m$ rows of $\bm X$ by $5$ to produce observations that are outlying with respect to the responses and also high leverage with respect to the predictors. The number of outliers $m$ is chosen from the grid $\{10, 20, 30, 50\}$. For the distance penalization, the ideal choice of the sparsity parameter $k$ is $5$. We employ five-fold cross-validation to select the tuning parameters for all three penalties. The sparsity level $k$ for distance penalization is varied over the grid $\{3,5,7,9, 11, 13, 15\}$, and the penalty constant $\rho$ is set to $10^8$ to enforce the desired sparsity as discussed in Section \ref{distance-penalization}. We initialize $\text{L}_2\text{E}$ estimation by setting $\bbeta_0 = \bm 0$ and $\eta_0 = -\log[\text{MAD}(\by)]$. All performance metrics depend on $100$ replicates. These metrics include: (a) estimation accuracy (measured by the relative error compared to the true $\bbeta$), (b) support recovery (measured by the F1 score), (c) the number of true positives, and (d) the number of false positives. The F1 score (harmonic mean of precision and recall) accounts for both true and false positives and takes on values in $[0,1]$, with a higher score indicating better support recovery.

\begin{figure}[t]
\centering
{\includegraphics[width=7cm]{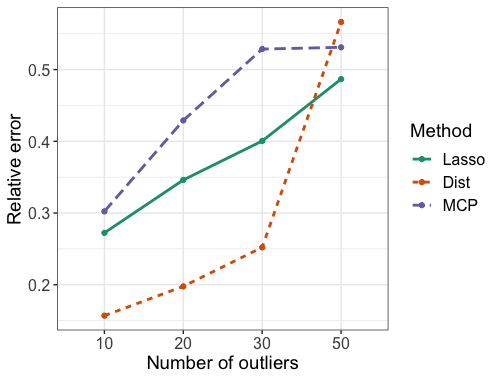}
}
\quad
{\includegraphics[width=7cm]{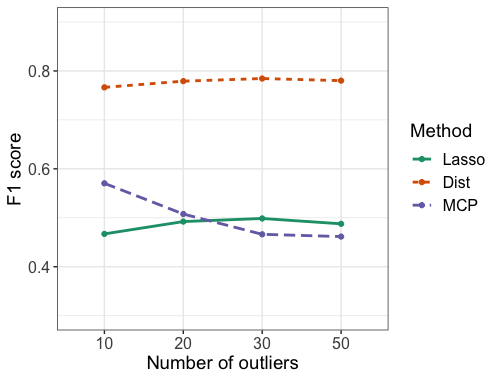}
}
{\includegraphics[width=7cm]{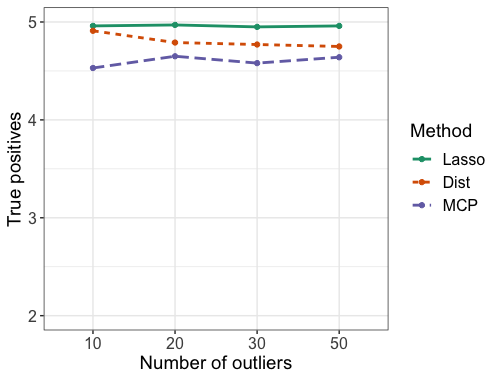}
}
\quad 
{\includegraphics[width=7cm]{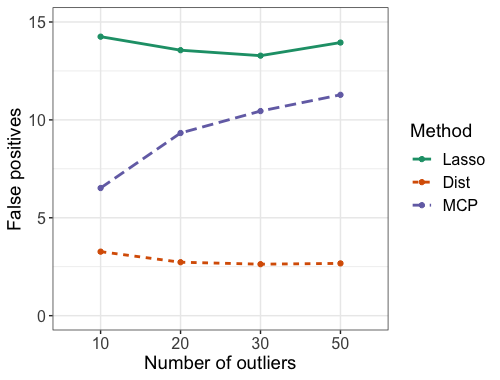}
}
  \setlength{\abovecaptionskip}{-5pt}
\caption{Simulation results for sparse regression under different numbers of outliers. Average performance based on $100$ replicates for each method.}
\label{fig: sparse}
\end{figure}

Figure \ref{fig: sparse} shows the performance of the Lasso,  MCP, and  distance penalties in robust sparse regression with the $\text{L}_2\text{E}$ loss under different numbers of outliers. Estimation degrades for all three methods as the number of outliers increases. Distance penalization consistently achieves a lower relative error than Lasso and MCP, except for $m=50$, where all methods produce unacceptable estimates. In support recovery, distance penalization consistently delivers a much higher F1 score than Lasso and MCP. The two plots in the bottom row of Figure \ref{fig: sparse} highlight the difference in support recovery among the three methods. Lasso identifies the most true positives but suffers from the most false positives in each scenario. MCP selects fewer irrelevant variables compared to Lasso but misses some true positives. In contrast, distance penalization identifies a number of true positives comparable to Lasso while maintaining a much lower false positive rate.

\begin{figure}[t]
\centering
{\includegraphics[width=7cm]{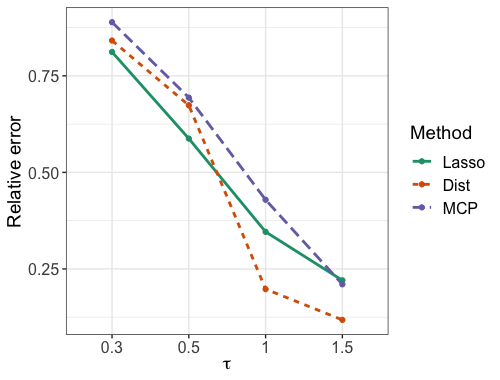}
}
\quad
{\includegraphics[width=7cm]{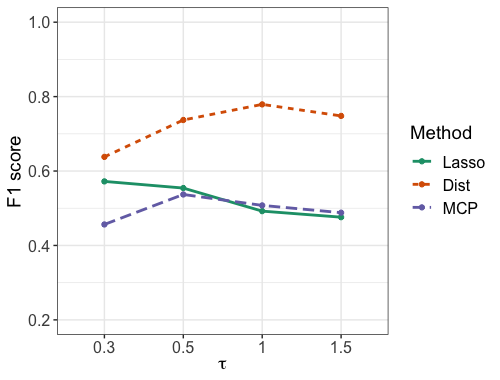}
}
{\includegraphics[width=7cm]{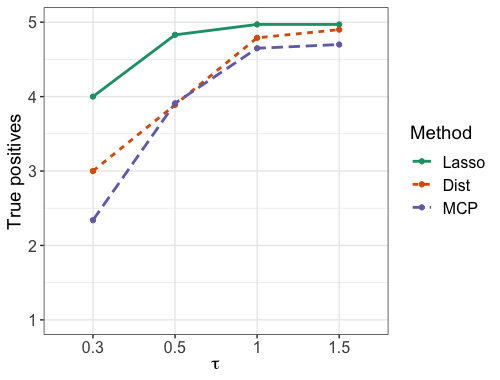}
}
\quad
{\includegraphics[width=7cm]{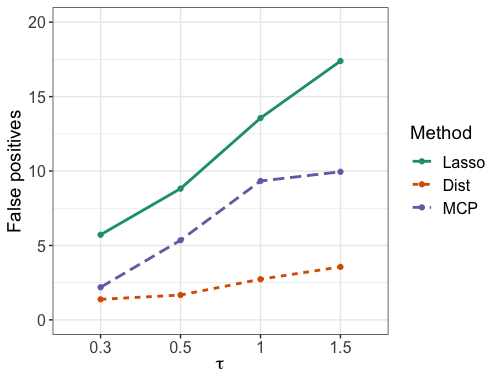}
}
  \setlength{\abovecaptionskip}{-10pt}
\caption{Simulation results for sparse regression under different noise levels. Average performance based on $100$ replicates for each method.}
\label{fig: sparse-snr}
\end{figure}

In the second experiment, we compare the performance of the different analysis methods (Lasso, MCP, and distance penalization) under different noise levels. We fix the number of outliers at $m=20$ and sample the precision parameter $\tau$ over the grid $\{0.3, 0.5, 1, 1.5\}$. A small value of $\tau$ represents a high noise level. We use the rules of our first experiment to produce outliers, select tuning parameters, and initialize $\text{L}_2\text{E}$ estimation. Figure \ref{fig: sparse-snr} summarizes our analysis results under different noise levels. As expected, the estimation errors of all three methods decrease as the value of $\tau$ increases. Distance penalization outperforms Lasso and MCP in estimation accuracy when the noise level is relatively low ($\tau \geq 1$). In addition, distance penalization compares favorably with Lasso and MCP in F1 score across different noise levels. The plots of true and false positives provide detailed insight into the support recovery of the different methods. All methods achieve a larger number of true positives as the value of $\tau$ increases, with Lasso leading the others. However, Lasso is plagued by an increasingly large number of false positives as the value of $\tau$ increases. Distance penalization achieves a smaller number of false positives, is less sensitive to the noise than Lasso and MCP, and stands out among the three methods in support recovery. This sparse regression example emphasizes the flexibility of $\text{L}_2\text{E}$ regression in accommodating different penalization methods and the advantages of distance penalization in both estimation accuracy and structure recovery.

\section{Real Data Application}
\label{real_data}

To illustrate the application of $\text{L}_2\text{E}$ regression in unconstrained robust multivariate regression and its effectiveness in detecting outliers, we now turn to the Hertzsprung-Russell diagram data of star cluster CYG OB1 investigated in \cite{rousseeuw2005robust, L2E-Scott, scott2021robust}. This data set includes two variables collected from $47$ stars in the direction of Cygnus. The predictor variable is the logarithm of the temperature at the star's surface, and the response variable is the logarithm of its light intensity. Though small, this data set is commonly used in robust regression owing to its four known outliers -- four bright giant stars observed at low temperatures \citep{star-data}. 

In this example, the penalty term $\phi(\bm \beta)=0$. Therefore, the MM update of $\bm \beta$ reduces to a standard least squares problem solvable by many efficient algorithms. In our implementation, we invoke the \texttt{lm} function in the R package \texttt{stats} \citep{R}. We initialize $\bbeta_0=\bm 0$ and  $\eta_0 = -\log[\text{MAD}(\by)]$. The left panel in Figure \ref{fig: real_data} displays the fitted $\text{L}_2\text{E}$ regression model. In comparison with ordinary least squares, $\text{L}_2\text{E}$ successfully reduces the influence of the four outliers and fits the remaining data points well. The converged weights $w_i = e^{-\tau^2r_i^2/2}$, where $r_i$ denotes the $i$-th $\text{L}_2\text{E}$ residual, serve as a diagnostic tool to detect outliers. As discussed in Section \ref{Sec4}, a small weight suggests a potential outlier. The histogram of the logarithm of weights in the right panel in Figure \ref{fig: real_data} clearly identifies the four outliers. These are colored in red in the scatter plot in the left panel. As a practical matter, we tried different initializations of $\bbeta$ in $\text{L}_2\text{E}$ estimation. Different initial values could potentially lead to different estimates. A direct and simple way to compare initializations is to rank their converged $\text{L}_2\text{E}$ losses \eqref{L2E-loss}. In this real data example, the neutral initialization $\bbeta_0=\bzero$ yields the smallest $\text{L}_2\text{E}$ loss.

\begin{figure}[t]
\centering
{\includegraphics[width=8cm]{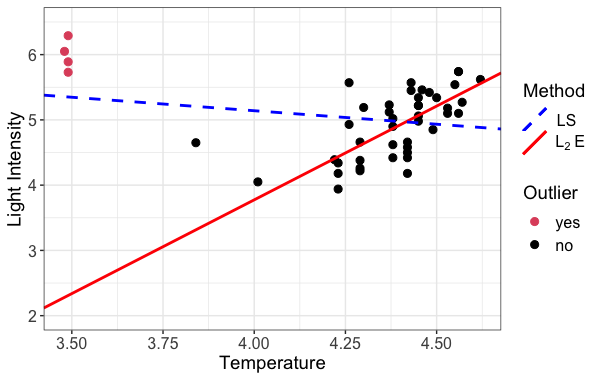}
}
{\includegraphics[width=6.65cm]{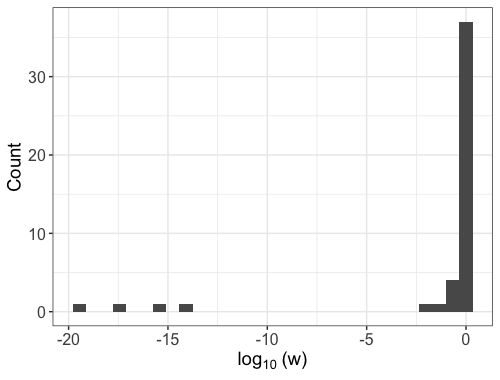}
}
  \setlength{\abovecaptionskip}{-10pt}
\caption{Fitted regression models from  $\text{L}_2\text{E}$ and LS for the Hertzsprung-Russell Diagram Data (left panel). The four known outliers are successfully detected by the $\text{L}_2\text{E}$ according to the histogram of the resulting weights (right panel).}
\label{fig: real_data}
\end{figure}

\section{Discussion}
\label{Sec6}

Because robust structured regression is resistant to the undue influence of outliers, it is valuable in many noisy data applications. The $\text{L}_2\text{E}$ computational framework \citep{L2E} for robust structured regression has the advantage of allowing the simultaneous estimation of regression coefficients and precision. 
This paper retains the overall strategy of block descent but introduces several non-trivial improvements. We introduce an MM algorithm based on a sharp majorization to accelerate convergence.  Each MM update of $\bbeta$ reduces to penalized least squares and can be readily handled by existing regression solvers. Although this plug-and-play tactic already formed part of the proximal gradient algorithm in \cite{L2E}, our tight majorization leads to better results. We also reparameterize  precision to avoid box constraint and update the new precision parameter by an approximate Newton's method. The computational cost per iterate remains the same, but again the number of iterations until convergence drops considerably. Finally, we extend penalization to distance and nonconvex penalties. These steps lead to  better statistical performance and model selection.

We demonstrate the merits of our refined computational framework through a rich set of simulation examples, including isotonic regression, convex regression, sparse regression, and trend filtering, and a real data  application to unconstrained multivariate regression. Given the same penalties, our simulation results show that the new algorithms outperform the original ones in both computational speed and estimation accuracy. Distance penalties to sparsity sets, in particular, show competitive advantages in both estimation accuracy and model selection. The real data example illustrates the convenience of using  the refined framework to identify outliers. Overall, the innovations introduced here make $\text{L}_2\text{E}$ an attractive tool for robust regression.


\section*{Supplementary Material}
Supplementary materials and code for this article are available online.
The supplement.pdf file contains the two simulation examples of convex regression and trend filtering under the $\text{L}_2\text{E}$ criterion. The L2E-code.zip file includes code for implementing the $\text{L}_2\text{E}$ isotonic regression and reproducing Figures 2 and 3 in the paper. To implement other $\text{L}_2\text{E}$ regression methods in the article, we refer readers to the eponymous \texttt{L2E} R package on the CRAN. 

\section*{Acknowledgement}

The authors are grateful to the editor, the associate editor, and the two referees for their helpful comments and suggestions. The authors thank Lisa Lin for her help with the R package.

\section*{Funding}
Lange's work is supported by the United States Public Health Service (USPHS) grants GM53275 and HG006139. Chi's work is partly supported by the National Science Foundation (NSF) grant DMS-2201136 and National Institutes of Health (NIH) grant R01GM135928.

\section*{Conflict of interest}

The authors report there are no competing interests to declare.







\bibliographystyle{jasa3}

\bibliography{main}

\end{document}